\documentclass{article}
\usepackage{fullpage}
\usepackage{graphicx}
\usepackage{epstopdf}
\usepackage{latexsym}
\usepackage{amsmath}
\usepackage{amsfonts}
\usepackage{amssymb}
\usepackage{mathrsfs}
\usepackage{pifont}
\usepackage{yhmath}
\usepackage{undertilde}
\usepackage{bbm}
\usepackage{dsfont}
\usepackage{eufrak}
\usepackage{amsthm}
\usepackage[usenames,dvipsnames]{color}
\usepackage{stmaryrd}
\usepackage{wasysym}
\usepackage{bbm}
\usepackage[ruled]{algorithm2e}

\newtheorem{theorem}{Theorem}[section]
\newtheorem{lemma}[theorem]{Lemma}
\newtheorem{corollary}[theorem]{Corollary}

\newtheorem{conjecture}[theorem]{Conjecture}

\theoremstyle{remark}

\newcommand{\redcomment}[1]{\textcolor{black}{\textrm{#1}}}

\begin{document}
\newcounter{my}
\newenvironment{mylabel}
{
\begin{list}{(\roman{my})}{
\setlength{\parsep}{-1mm}
\setlength{\labelwidth}{8mm}
\usecounter{my}}
}{\end{list}}

\newcounter{my2}
\newenvironment{mylabel2}
{
\begin{list}{(\alph{my2})}{
\setlength{\parsep}{-0mm} \setlength{\labelwidth}{8mm}
\setlength{\leftmargin}{3mm}
\usecounter{my2}}
}{\end{list}}

\newcounter{my3}
\newenvironment{mylabel3}
{
\begin{list}{(\alph{my3})}{
\setlength{\parsep}{-1mm}
\setlength{\labelwidth}{8mm}
\setlength{\leftmargin}{10mm}
\usecounter{my3}}
}{\end{list}}

\title{\bf Sufficient Conditions for Tuza's Conjecture on\\ Packing and Covering Triangles\thanks{Research supported in part  by NNSF of China under Grant No. 11531014 and 11222109, and by CAS Program for Cross \& Cooperative Team of Science  \& Technology Innovation.}}
\date{}
\maketitle

\vspace{-5em}

\begin{center}\large

\author{Xujin Chen \quad Zhuo Diao \quad Xiaodong Hu \quad Zhongzheng Tang }

\vspace{+1em}

 Academy of Mathematics and Systems Science \\ Chinese Academy of Sciences, Beijing 100190, China

\vspace{+1em}

\{xchen,diaozhuo,xdhu,tangzhongzheng\}@amss.ac.cn

\end{center}


\begin{abstract}
Given a simple graph $G=(V,E)$, a subset of $E$ is called a triangle cover if it intersects each triangle of $G$.
Let $\nu_t(G)$ and  $\tau_t(G)$ denote the maximum number of pairwise edge-disjoint triangles in $G$
and the minimum cardinality of a triangle cover of $G$, respectively.
Tuza conjectured in 1981 that  $\tau_t(G)/\nu_t(G)\le2$ holds for every graph $G$.
In this paper, using a hypergraph approach, we design polynomial-time combinatorial algorithms
for finding small triangle covers.
These algorithms imply new sufficient conditions for Tuza's conjecture on covering and packing triangles.
More precisely, suppose that the set $\mathscr T_G$ of triangles covers all edges in $G$. We
  show that  a triangle cover of $G$ with cardinality at most $2\nu_t(G)$ can be found in polynomial time
if one of the following conditions is satisfied:
(i) $\nu_t(G)/|\mathscr T_G|\ge\frac13$, (ii) $\nu_t(G)/|E|\ge\frac14$, (iii) $|E|/|\mathscr T_G|\ge2$.\\

\noindent{\bf Keywords:} Triangle cover, Triangle packing, Linear 3-uniform hypergraphs, Combinatorial algorithms

\end{abstract}

\section{Introduction}

Graphs considered in this paper are undirected, simple and finite (unless otherwise noted). Given a  graph $G=(V,E)$ with vertex set $V(G)=V$ and edge set $E(G)=E$, for convenience, we often identify a triangle in $G$ with its edge set. A subset   of $E$ is called a {\em triangle cover} if {it intersects} each triangle of $G$. Let $\tau_t(G)$ denote the minimum cardinality of a triangle cover of $G$, referred to as the {\em triangle covering number} of $G$. A set  of pairwise edge-disjoint triangles {in $G$} is called a {\em triangle packing} of $G$. Let $\nu_t(G)$ denote the maximum cardinality of a triangle packing of $G$, referred to as the {\em triangle  packing number} of $G$. It is clear that  $1\leq\tau_t(G)/\nu_t(G)\le3$ holds for every graph $G$. Our research is motivated by the following  conjecture {raised by   Tuza \cite{tuza1981} in 1981}.

\begin{conjecture}[Tuza's Conjecture \cite{tuza1981}] \label{coj:tuza}
 $\tau_t(G)/\nu_t(G)\le2$ holds for every graph $G$.
\end{conjecture}

To the best of our knowledge, the conjecture is still {unsolved in general}. If {it is true}, then the upper bound 2   is sharp as shown by $K_4$ and $K_5$ -- the complete graphs of orders $4$ and $5$.

  \paragraph{Related work.} The only known universal upper bound smaller than 3 was given by Haxell \cite{haxell19}, who shown that $\tau_t(G)/\nu_t(G)\le 66/{23} =2.8695...$ holds for all graphs $G$. {Haxell's proof  \cite{haxell19} implies a polynomial-time algorithm for finding a triangle cover of cardinality at most 66/23 times that of a maximal triangle packing}. Other partial results on {Tuza's} conjecture concern with special classes of graphs.

Tuza \cite{tuza1990} proved {his} conjecture holds for planar graphs, $K_{5}$-free chordal graphs and graphs with $n$ vertices and at least $7n^{2}/16$ edges. {The proof for planar graphs \cite{tuza1990} gives an elegant polynomial-time algorithm for finding a triangle cover in planar graphs with cardinality at most twice that of a maximal triangle packing. The validity of Tuza's conjecture on the}  class of planar graphs was later {generalized} by Krivelevich \cite{Krivelevich1995}   to the class of graphs without $K_{3,3}$-subdivision.    Haxell and Kohayakawa \cite{HK1998} showed that $\tau_t(G)/\nu_t(G)\le2-\epsilon$ for tripartite graphs $G$, where $\epsilon > 0.044$.  Haxell, Kostochka and Thomasse \cite{HKT2012} proved that every $K_{4}$-free planar graph $G$ satisfies $\tau_t(G)/\nu_t(G)\le1.5$. 

Regarding the tightness of the {conjectured upper bound 2,   Tuza \cite{tuza1990} noticed that infinitely many graphs $G$ attain the conjectured upper bound $\tau_t(G)/\nu_t(G)=2$.} Cui,  Haxell and  Ma \cite{CHM2009} characterized   planar graphs $G$ satisfying $\tau_t(G)/\nu_t(G)=2${; these graphs are edge-disjoint unions of $K_4$'s} plus possibly some vertices and edges that are not in triangles. Baron and  Kahn \cite{BK2014} proved that Tuza's conjecture is asymptotically tight for dense graphs.

Fractional and weighted variants of Conjecture \ref{coj:tuza} were studied in literature. Krivelevich \cite{Krivelevich1995} {proved} two fractional versions of the conjecture: $\tau_t(G)\leq2\nu^{\ast}_t(G)$ and $\tau^{\ast}_t(G)\leq2\nu_t(G)${, where $\tau^{\ast}_t(G)$ and $\nu^{\ast}_t(G)$ are the values of an optimal fractional triangle cover and an optimal fractional triangle packing of $G$, respectively. The result was generalized by Chapuy  et al. \cite{CDDMS2015} to the weighted version}, which amounts to packing and covering triangles in multigraphs {$G_w$ (obtained from $G$ by adding multiple edges). The authors \cite{CDDMS2015} showed that  $\tau(G_w)\leq2\nu^{\ast}(G_w)-\sqrt{\nu^{\ast}(G_w)/6}+1$ and $\tau^{\ast}(G_w)\leq2\nu(G_w)$; the arguments imply an LP-based 2-approximation algorithm for finding a minimum weighted triangle cover.}

\paragraph{Our contributions.} Along a different line, we establish new sufficient conditions for validity of  Tuza's conjecture by comparing the {triangle} packing number, the number of triangles and the number of edges. {Given a graph $G$, we use $\mathscr T_G=\{E(T):T$ is a triangle in $G\}$ to denote the set consisting of the (edge sets of) triangles in $G$.} Without loss of generality, we focus on the graphs in which every edge is contained in some triangle. These graphs are called {\em irreducible}.
\begin{theorem}\label{th:condition}
Let  $G=(V,E)$ be an irreducible graph. Then a triangle cover of $G$ with cardinality at most $2\nu_t(G)$ can be found in polynomial time, which implies $\tau_t(G)\le2\nu_t(G)$, if one of the following conditions is satisfied:
(i) $\nu_t(G)/|\mathscr T_G|\ge\frac13$, (ii) $\nu_t(G)/|E|\ge\frac14$, (iii) $|E|/|\mathscr T_G|\ge2$.
\end{theorem}
The primary idea behind the theorem is simple: any one of conditions (i) -- (iii) allows us to remove at most $\nu_t(G)$ edges from $G$ to make the resulting graph $G'$ satisfy $\tau_t(G')=\nu_t(G')$; the removed edges and the edges in a minimum triangle cover of  $G'$ form a triangle cover of $G$ with size at most $\nu_t(G)+\nu_t(G')\le2\nu_t(G)$. The idea is realized by establishing new results on linear 3-uniform hypergraphs (see Section \ref{sec:hypergraph}); the most important one states that such a hypergraphs could be made acyclic by removing {a number of vertices that is no more than a third of the number of its edges}. A key observation here is that hypergraph  $(E,\mathscr T_G)$ is linear and 3-uniform.


To show the qualities of conditions (i) -- (iii) in Theorem \ref{th:condition}, we obtain the following result  which complements to the constants $\frac13$, $\frac14$ and 2 in these conditions with $\frac14$, $\frac15$ and $\frac32$, respectively.
\begin{theorem}\label{th:howgood}
Tuza's conjecture holds for every graph if there exists some real $\delta>0$ such that Tuza's conjecture holds for every irreducible graph $G$ satisfying one of the following properties: (i') $\nu_t(G)/|\mathscr T_G|\ge \frac14-\delta$, (ii') $\nu_t(G)/|E|\ge\frac15-\delta$, (iii') $|E|/|\mathscr T_G|\ge\frac32-\delta $.
\end{theorem}
We also investigate Tuza's conjecture on classical Erd\H{o}s-R\'{e}nyi random graph $\mathcal G(n,p)$, and prove that $\text{\bf Pr}[\tau_t(G)/\nu_t(G)\le2]=1-o(1)$ provided $G\in\mathcal G(n,p)$ and $p> \sqrt{3}/2$.

\redcomment{It is worthwhile pointing out that strengthening Theorem \ref{th:condition}, our arguments actually establish stronger results for linear 3-uniform hypergraphs (see Theorem \ref{th:condition'})}.

\medskip
{The rest of paper is organized as follows. Section \ref{sec:hypergraph} proves theoretical and algorithmic results on linear 3-uniform hypergraphs concerning feedback sets, which are main technical tools for establishing new sufficient conditions for Tuza's conjecture in Section \ref{sec:tuza}. Section \ref{sec:conclude} concludes the paper with extensions and future research directions.}

\section{Hypergraphs}\label{sec:hypergraph}

{This section develops hypergraph tools for studying Tuza's conjecture. The theoretical and algorithmic results are of interest in their own right.}

Let $\mathcal H=(\mathcal V,\mathcal E)$ be a hypergraph with vertex set $\mathcal V$ and edge set $\mathcal E$. For convenience, we use $|\!|\mathcal H|\!|$ to denote the number $|\mathcal E|$ of edges in $\mathcal H$.  If hypergraph $\mathcal H'=(\mathcal V',\mathcal E')$ satisfies $\mathcal V' \subseteq \mathcal V$  and $ \mathcal E' \subseteq \mathcal E$, we call $\mathcal H'$ a {\em sub-hypergraph} of $\mathcal H$, and write $\mathcal H'\subseteq \mathcal H$. For each $v\in \mathcal V$, the {\em degree} $d_{\mathcal H}(v)$ is the number of edges in $\mathcal E$ that {contain} $v$. We say $v$ is an {\em isolated vertex} of $\mathcal H$ if $d_{\mathcal H}(v)=0$. Let $k\in\mathbb N$ be a positive integer, hypergraph $\mathcal H$ is called  {\em $k$-regular}  if $d_{\mathcal H}(u)=k$ for each $u\in \mathcal V$, and   {\em $k$-uniform} if $|e|=k$ for each $e\in \mathcal E$. Hypergraph $\mathcal H$ is   {\em linear} if $|e\cap f|\le1$ for any pair of distinct edges $e,f\in \mathcal E$.

A vertex-edge alternating sequence $ v_{1}e_{1}v_{2}...v_{k}e_{k}v_{k+1}$ of $\mathcal H$ is called a {\em path} (of {\em length} $k$) between $v_{1}$ and $v_{k+1}$ if $v_{1}, v_{2},..., v_{k+1}\in\mathcal V$ are distinct, $ e_{1}, e_{2},..., e_{k}\in\mathcal E$ are distinct, and  $\{v_{i},v_{i+1}\}\subseteq e_{i}$ for each $i\in [k]=\{1,\ldots,k\}$. We consider each vertex of $\mathcal H$ as a path of length 0.  Hypergraph $\mathcal H$ is said to be {\em connected} if  there is a path  between any pair of distinct vertices in $\mathcal H$. A maximal connected sub-hypergraph of $\mathcal H$  is called a {\em component} of $\mathcal H$. Obviously, $\mathcal H$ is connected if and only if it has only one component.

  A vertex-edge alternating sequence   $\mathcal C= v_{1}e_{1}v_{2}e_{2}...v_{k}e_{k}v_{1}$, where $k\ge2$, is called a {\em cycle} (of length $k$) if $v_{1}, v_{2},..., v_{k}\in\mathcal V$ are distinct, $ e_{1}, e_{2},..., e_{k}\in \mathcal E$ are distinct, and  $\{v_{i},v_{i+1}\}\subseteq e_{i}$ for each $i\in [k]$, where $v_{k+1}=v_{1}$. We consider the cycle $\mathcal C$ as a sub-hypergraph of $\mathcal H$ with vertex set $\cup_{i\in[k]}e_i$ and edge set {$\{e_{i}: i\in [k]\}$}. For any $\mathcal S\subset\mathcal V$ (resp. $S\subset\mathcal E$), we write $\mathcal H\setminus \mathcal S$  for the  sub-hypergraph of $\mathcal H$ obtained from $\mathcal H$ by deleting all vertices in $\mathcal S$ and all edges  incident with some vertices in $\mathcal S$ (resp. deleting all edges in $\mathcal E$ and keeping vertices). If $\mathcal S$ is a singleton set $\{s\}$, we write $\mathcal H\setminus s$ instead of $\mathcal H\setminus \{s\}$. For any $\mathcal S\subseteq2^{\mathcal V}$, the hypergraph  $(\mathcal V,\mathcal E\cup \mathcal S)$ is often written as $\mathcal H\uplus\mathcal E$, and as $\mathcal H\oplus\mathcal S$ if $\mathcal S\cap\mathcal E=\emptyset$.

A vertex (resp. edge) subset  of $\mathcal H$ is called a {\em feedback vertex set} or FVS (resp. {\em feedback edge set} or FES)  of $\mathcal H$ if it intersects the vertex (resp. edge) set  {of} every cycle of $\mathcal H$. A vertex subset of $\mathcal H$ is called a {\em transversal} of $\mathcal H$ if it intersects every edge of $\mathcal H$. Let  $\tau^{{}_{\mathcal V}}_c(\mathcal H)$,  $\tau_c^{{}_{\mathcal E}}(\mathcal H)$ and $\tau(\mathcal H)$ denote, respectively, the minimum cardinalities of a FVS, a FES, and a transversal  of $\mathcal H$. A {\em matching} of $\mathcal H$ is an nonempty set of pairwise disjoint edges of $\mathcal H$. Let $\nu(\mathcal H)$
denote the maximum cardinality of a matching of $\mathcal H$. It is easy to see that $\tau^{{}_{\mathcal V}}_c(\mathcal H)\leq \tau_c^{{}_{\mathcal E}}(\mathcal H)$, $\tau^{{}_{\mathcal V}}_c(\mathcal H)\leq \tau(\mathcal H)$ and $ \nu(\mathcal H)\le\tau(\mathcal H)$. {Our discussion will frequently use the trivial observation that if no cycle of $\mathcal H$ contains any element of some subset $\mathcal S$ of $\mathcal V\cup\mathcal E$, then $\mathcal H$ and $\mathcal H\setminus S$ have the same set of FVS's, and $\tau_c^{{}_{\mathcal V}}(\mathcal H)= \tau_c^{{}_{\mathcal V}}(\mathcal H\setminus \mathcal S)$. The following theorem is one of main contributions of this paper.}

\begin{theorem}\label{1/3}
Let  $\mathcal H$ be a linear $3$-uniform hypergraph. Then $\tau_c^{{}_{\mathcal V}}(\mathcal H)\le |\!|\mathcal H|\!|/3$.
\end{theorem}
\begin{proof}  Suppose that the theorem failed. {We} take a counterexample $\mathcal H=(\mathcal V,\mathcal E)$ with $\tau_c^{{}_{\mathcal V}}(\mathcal H)> |\mathcal E|/3$ such that $|\!|\mathcal H|\!|=|\mathcal E|$ is as small as possible.
Obviously $|\mathcal E|\geq 3$. Without loss of generality, we can assume that $\mathcal H$ has no isolated vertices. Since $\mathcal H$ is linear, any cycle in $\mathcal H$ is of length at least 3.

If there exists $e\in \mathcal E$ which does not belong to any cycle of $\mathcal H$, then $\tau_c^{{}_{\mathcal V}}(\mathcal H)= \tau_c^{{}_{\mathcal V}}(\mathcal H\setminus e)$. The minimality of $\mathcal H=(\mathcal V,\mathcal E)$ implies $\tau_c^{{}_{\mathcal V}}(\mathcal H\setminus e)\leq (|\mathcal E|-1)/3$, giving $\tau_c^{{}_{\mathcal V}}(\mathcal H)< |\mathcal E|/3$, a contradiction. So we have
\begin{itemize}
\item[(1)]  Every edge in $\mathcal E$ is contained in some cycle of $\mathcal H$.
\end{itemize}

If there exists $v\in\mathcal V$ with $d_{\mathcal H}(v)\geq3$, then $\tau_c^{{}_{\mathcal V}}(\mathcal H\setminus v)\leq (|\mathcal E|-d_{\mathcal H}(v))/3\leq (|\mathcal E|-3)/3$, where the first inequality is due to the minimality of $\mathcal H$. Given a minimum FVS $\mathcal S$ of $\mathcal H\setminus v$, it is clear that $\mathcal S\cup \{v\}$ is a FVS of $\mathcal H$ with size $  |\mathcal S|+1=\tau_c^{{}_{\mathcal V}}(\mathcal H\setminus v)+1\leq |\mathcal E|/3$, a contradiction to $\tau_c^{{}_{\mathcal V}}(\mathcal H)> |\mathcal E|/3$. So we have

\begin{itemize}
\item[(2)] $ d_{\mathcal H}(v)\leq 2$ for all $v\in \mathcal V$.
\end{itemize}

Suppose that there exists $v\in\mathcal V$ with $d_{\mathcal H}(v)= 1$. Let $e_1\in\mathcal E$ be the unique edge that contains $v$. Recall   from (1) that $e_1$ is contained in a cycle  $\mathcal C= v_{1}e_{1}v_{2}e_{2}v_{3}\cdots e_{k}v_{1}$, where $k\ge3$. By (2), we have $d_{\mathcal H}(v_i)=2$ for all $i\in[k]$. In particular $d_{\mathcal H}(v_1)=d_{\mathcal H}(v_2)=2>d_{\mathcal H}(v)$ implies $v\not\in\{v_1,v_2\}$, and in turn $ v_1,v_2,v \in e_1$ enforces $e_{1}= \{v_{1},v,v_{2}\}$. Let $\mathcal S$ be a minimum FVS  of $\mathcal H'=\mathcal H\setminus\{e_1,e_2,e_3\}$. It follows from (2) {that
\[\mathcal H\setminus v_3\subseteq  \mathcal H\setminus\{e_2,e_3\}= \mathcal H'\oplus e_1,\]
and} in $\mathcal H'\oplus e_1$, {edge} $e_1$ intersects at most one other edge, and therefore is not contained by any cycle. 
Thus {$\mathcal S$ is a FVS of $ \mathcal H'\oplus e_1$, and hence a FVS of $\mathcal H\setminus v_3$, implying} that $\{v_3\}\cup\mathcal  S$ 
is a FVS of $\mathcal H$. We deduce that $|\mathcal E|/3<\tau_c^{{}_{\mathcal V}}(\mathcal H)\le|\{v_3\}\cup \mathcal S|\le 1+|\mathcal S|$. Therefore $\tau_c^{{}_{\mathcal V}}(\mathcal H')=|\mathcal S|>(|\mathcal E|-3)/3=|\!|\mathcal H'|\!|/3$ shows  a contradiction to the minimality of $\mathcal H$. Hence the vertices of $\mathcal H$ all have degree at least 2, which together with (2) gives

\begin{itemize}
\item[(3)] $\mathcal H$ is $2$-regular.
\end{itemize}

Let $\mathcal C=(\mathcal V_c,\mathcal E_c)=v_1e_1v_2e_2\ldots v_ke_kv_1$ be a shortest cycle in $\mathcal H$, where $k\ge3$. For each $i\in[k]$, suppose that $e_i=\{v_i,u_i,v_{i+1}\}$, where $v_{k+1}=v_1$.    


 Because $\mathcal C$ is a shortest cycle, for each   pair of distinct indices $i,j\in [k]$, we have $e_{i}\cap e_{j}=\emptyset$ if and only if $e_i$ and $e_j$ are not adjacent in $\mathcal C$, i.e., $|i-j|\not\in\{1,k-1\}$. This fact along with the linearity of $\mathcal H$ says that $v_1,v_2,\ldots,v_k,u_1,u_2,\ldots,u_k$ are distinct. By (3), each $u_i$ is contained in a unique edge $f_i\in\mathcal E\setminus\mathcal E_c$, $i\in[k]$. We distinguish among three cases  depending on the values of $k\pmod 3$. In each case, we construct a proper sub-hypergraph $\mathcal H'$ of $\mathcal H$ with $|\!|\mathcal H'|\!|<|\!|\mathcal H|\!|$ and $\tau_c^{{}_{\mathcal V}}(\mathcal H')>|\!|\mathcal H'|\!|/3$ which shows a contradiction to the minimality of $\mathcal H$.

\paragraph{\sc Case 1. $k\equiv0\pmod3$:} Let $\mathcal S$ be a minimum FVS  of $\mathcal H'=\mathcal H\setminus\mathcal E_c$. Setting {$\mathcal V_*=\{v_i:i\equiv0\pmod3, i\in [k]\}$ and $\mathcal E_*=\{e_i:i\equiv1\pmod3, i\in [k]\}$}, it follows from (3) that
\[\mathcal H\setminus \mathcal V_*\subseteq (\mathcal H\setminus\mathcal E_c)\oplus \mathcal E_*= \mathcal H'\oplus \mathcal E_*,\]
and in $\mathcal H'\oplus \mathcal E_*$, each edge in $\mathcal E_*$    intersects exactly one other edge, and therefore is not contained by any cycle. 
Thus $(\mathcal H'\oplus \mathcal E_*)\setminus \mathcal S$ is also acyclic, so is $(\mathcal H\setminus\mathcal V_*)\setminus\mathcal S$, saying that $\mathcal V_*\cup\mathcal  S$ 
is a FVS of $\mathcal H$. We deduce that $|\mathcal E|/3<\tau_c^{{}_{\mathcal V}}(\mathcal H)\le|\mathcal V_*\cup \mathcal S|\le k/3+|\mathcal S|$. Therefore $\tau_c^{{}_{\mathcal V}}(\mathcal H')=|\mathcal S|>(|\mathcal E|-k)/3=|\!|\mathcal H'|\!|/3$ shows  a contradiction.

  \paragraph{\sc Case 2. $k\equiv1\pmod3$:} Consider the case where $f_1\ne f_3$ or $f_2\ne f_4$. Relabeling the vertices and edges if necessary, we may assume
 without loss of generality that $f_1\ne f_3$. Let $\mathcal S$ be a minimum FVS   of $\mathcal H'=\mathcal H\setminus(\mathcal E_c\cup \{f_1,f_3\})$. Set $\mathcal V_*=\emptyset$, $\mathcal E_*=\emptyset$ if $k=4$ and {$\mathcal V_*=\{v_i:i\equiv0\pmod3, i\in [k]-[3]\}$, $\mathcal E_*=\{e_i:i\equiv1\pmod3, i\in [k]-[6]\}$} otherwise. In any case we have $|\mathcal V_*|=(k-4)/3$ and
  \[\mathcal H\setminus(\{u_1,u_3\}\cup\mathcal V_*)\subseteq (\mathcal H\setminus(\mathcal E_c\cup \{f_1,f_3\}))\oplus( \{e_2,e_4\}\cup\mathcal E_*)= \mathcal H'\oplus (\{e_2,e_4\}\cup\mathcal E_*).\]
Note from (3) that in  $\mathcal H'\oplus (\{e_2,e_4\}\cup\mathcal E_*)$, each edge in $ \{e_2,e_4\}\cup\mathcal E_*$ can intersect at most one other edge, and therefore is not contained by any cycle. 
  Thus $( \mathcal H'\oplus (\{e_2,e_4\}\cup\mathcal E_*))\setminus \mathcal S$ is also acyclic, so is $(\mathcal H\setminus (\{u_1,u_3\}\cup\mathcal V_*))\setminus\mathcal  S$. Thus $\{u_1,u_3\}\cup\mathcal V_*\cup\mathcal S$ is a FVS of $\mathcal H$, and $|\mathcal E|/3<\tau_c^{{}_{\mathcal V}}(\mathcal H)\le|\{u_1,u_3 \}\cup\mathcal V_*\cup\mathcal S|\le2+|\mathcal V_*|+|\mathcal S|=(k+2)/3+|\mathcal S|$. This gives $\tau_c^{{}_{\mathcal V}}(\mathcal H')=|\mathcal S|>(|\mathcal E|-k-2)/3=|\mathcal H'|/3$, a contradiction.

Consider the case where $f_1= f_3$ and $f_2=f_4$. As $u_1,u_2,u_3,u_4$ are distinct and $|f_1|=|f_2|=3$, we have $f_1\ne f_2$. Observe that $u_1e_1v_2e_2v_3e_3u_3f_3u_1$ is a cycle in $\mathcal H$ of length 4. The minimality of $k$ enforces $k=4$.  Therefore $\mathcal E_c\cup\{f_1,f_2\}$ consist of 6 distinct edges. Let $\mathcal S$ be a minimum FVS of $\mathcal H'=\mathcal H\setminus(\mathcal E_c\cup\{f_1,f_2\})$. It follows from (3) that \[\mathcal H\setminus \{u_2,u_4\}\subseteq (\mathcal H\setminus (\mathcal E_c\cup\{f_1,f_2\}))\oplus\{e_1,e_3,f_1\}= \mathcal H'\oplus\{e_1,e_3,f_1\}.\]
In $\mathcal H'\oplus\{e_1,e_3,f_1\}$, both $e_{1}$ and $e_{3}$ intersect only one other edge, which is $f_1$, and any cycle through {$f_1$} must contain $e_1$ or $e_3$. It follows that none of $e_1,e_3,f_1$ is contained by a cycle of $\mathcal H'\oplus\{e_1,e_3,f_1\}$.  Thus $(\mathcal H'\oplus \{e_1,e_3,f_1\})\setminus \mathcal S$ is  acyclic, so is $(\mathcal H \setminus \{u_2,u_4\})\setminus \mathcal S$, saying that $ \{u_2,u_4\}\cup\mathcal S$ is a FVS of $\mathcal H$. Hence $|\mathcal E|/3<\tau_c^{{}_{\mathcal V}}(\mathcal H)\le|\{u_2,u_4 \}\cup\mathcal S|\leq2+ |\mathcal S|$. In turn $\tau_c^{{}_{\mathcal V}}(\mathcal H')=|\mathcal S|>(|\mathcal E|-6)/3=|\!|\mathcal H'|\!|/3$ shows a contradiction.

 \paragraph{\sc Case 3. $k\equiv2\pmod3$:} Let $\mathcal S$ be  a minimum FVS of $\mathcal H'=\mathcal H\setminus (\mathcal E_c\cup\{f_1\})$. Setting  {$\mathcal V_*=\{v_i:i\equiv1\pmod3, i\in [k]-[3]\}$ and $\mathcal E_*=\{e_i:i\equiv2\pmod3, i\in [k]\}$}, we have $|\mathcal V_*|=(k-2)/3$ and
  \[\mathcal H\setminus(\{u_1\}\cup\mathcal V_*)\subseteq (\mathcal H\setminus(\mathcal E_c\cup \{f_1\}))\oplus  \mathcal E_*= \mathcal H'\oplus\mathcal E_*\]
  In $\mathcal H'\oplus\mathcal E_*$, each edge in $\mathcal E_*$ intersects at most one other edge, and therefore is not contained by any cycle.   Thus $(\mathcal H'\oplus\mathcal E_*)\setminus \mathcal S$ is  acyclic, so is $(\mathcal H\setminus ( \{u_1\}\cup\mathcal V_*))\setminus\mathcal S$. Hence $\{u_1\}\cup\mathcal V_*\cup\mathcal S$ is a FVS of $\mathcal H$, yielding $|\mathcal E|/3<\tau_c^{{}_{\mathcal V}}(\mathcal H)\le| \{u_1\}\cup\mathcal V_*\cup\mathcal S|\le1+(k-2)/3+|\mathcal S|$ and a contradiction $\tau_c^{{}_{\mathcal V}}(\mathcal H')=|\mathcal S|>(|\mathcal E|-k-1)/3=|\!|\mathcal H'|\!|/3$.

\medskip The combination of the above three cases complete the proof.
  \end{proof}

{We remark that the upper bound $|\!|\mathcal H|\!|/3$ in Theorem \ref{1/3} is best possible. See Figure~\ref{sharp} for illustrations of five 3-uniform linear hypergraphs attaining the upper bound.  It is easy to prove that the maximum degree of every extremal hypergraph (those $\mathcal H$ with $\tau_c^{{}_{\mathcal V}}(\mathcal H)= |\!|\mathcal H|\!|/3$) is at most three. It would be interesting to characterize all extremal hypergraphs for Theorem \ref{1/3}.}

\begin{figure}[sharp]
\begin{center}
\includegraphics[scale=0.65]{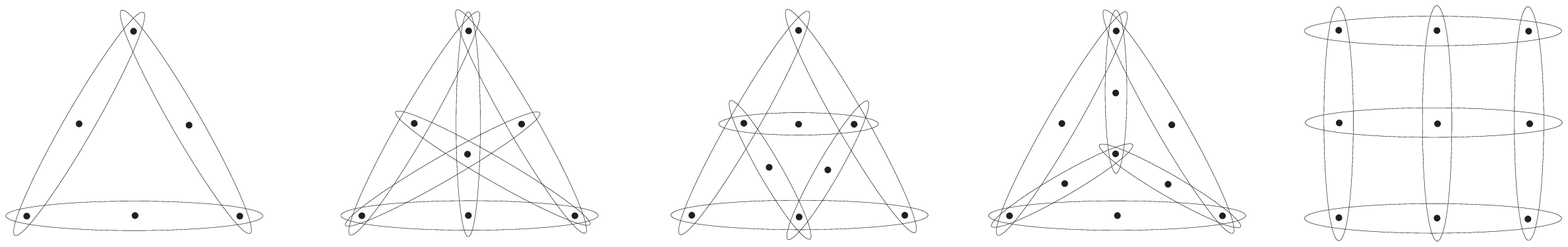}
\caption{\label{sharp}{Extremal linear 3-uniform hypergraphs $\mathcal H$ with $\tau_c^{{}_{\mathcal V}}(\mathcal H)= |\!|\mathcal H|\!|/3$.}}
\end{center}
\end{figure}

\medskip
The  proof of Theorem \ref{1/3} actually gives a recursive combinatorial algorithm for finding in polynomial time a FVS of size at most $|\!|\mathcal H|\!|/3$ on a linear 3-uniform hypergraph $\mathcal H$.
  \begin{algorithm} 
\KwIn{Linear 3-uniform hypergraph $\mathcal H=(\mathcal V,\mathcal E)$.}
\KwOut{$\textsc{Alg1}(\mathcal H)$, which is {a FVS   of  $\mathcal H$ with cardinality at most $|\!|\mathcal H|\!|/3$}.}
\begin{mylabel}
   \vspace{1.5mm}\item[1.]  \textbf{If} $|\mathcal E|\le2$ \textbf{Then} $\textsc{Alg1}(\mathcal H)\leftarrow\emptyset$
     \vspace{0.mm}\item[2.] \quad\textbf{Else} \textbf{If} {$\exists$ $s\in \mathcal V\cup\mathcal E$ such that  $s$ is not contained in any cycle of $\mathcal H$}
  \vspace{0.mm}\item[3.] \quad\quad\quad\hspace{3mm} \textbf{Then} $\textsc{Alg1}(\mathcal H)\leftarrow \textsc{Alg1}(\mathcal H\setminus s)$

  \vspace{0.5mm}\item[4.] \quad\quad\quad\hspace{.7mm}  \textbf{If} $\exists$ $s\in \mathcal V$ {such that} $d_{\mathcal H}(s)\ge3$
  \vspace{0.mm}\item[5.] \quad\quad\quad\hspace{3mm} \textbf{Then} $\textsc{Alg1}(\mathcal H)\leftarrow\{s\}\cup\textsc{Alg1}(\mathcal H\setminus s)$
   \vspace{0.5mm}\item[6.] \quad\quad\quad\hspace{.7mm} \textbf{If} $\exists$ $v\in\mathcal V$ such that $d_{\mathcal H}(v)=1$
  \vspace{0.mm}\item[7.] \quad\quad\quad\hspace{3mm} \textbf{Then} Let $v_1e_1v_2e_2v_3\cdots e_kv_1$ be a cycle of $\mathcal H$ such that $e_1=\{v_1,v_2,v\}$
  \vspace{0.mm}\item[8.] \quad\quad\quad\quad\quad\hspace{5mm} {$\textsc{Alg1}(\mathcal H)\leftarrow\{v_3\}\cup\textsc{Alg1}(\mathcal H\setminus \{e_1,e_2,e_3\})$}
   \vspace{0.mm}\item[9.] \quad\quad\quad\hspace{0.5mm} Let $(\mathcal V_c,\mathcal E_c)=v_1e_1v_2e_2\ldots v_ke_kv_1$ be a shortest cycle in $\mathcal H$
     \vspace{0.mm}\item[10.] \quad\quad\quad\hspace{0.5mm} For each $i\in[k]$, let $u_i\in\mathcal V_c$, $f_i\in\mathcal E\setminus\mathcal E_c$ be such that $\{u_i,v_i,v_{i+1}\}=e_i$, $u_i\in f_i$
     \vspace{0.mm}\item[11.] \quad\quad\quad\hspace{.7mm} \textbf{If} $k\equiv0\pmod3$ \textbf{Then} {$\textsc{Alg1}(\mathcal H)\leftarrow\{v_i: i\equiv0\pmod3, i\in [k]\}\cup\textsc{Alg1}(\mathcal H\setminus \mathcal E_c)$}
  \vspace{0.mm}\item[12.] \quad\quad\quad\hspace{.7mm} \textbf{If} $k\equiv1\pmod3$
 \vspace{0.mm}\item[13.] \quad\quad\quad\hspace{3mm}  \textbf{Then} \textbf{If} $f_1\ne f_3$ or $f_2\ne f_4$
  \vspace{0.mm}\item[14.] \quad\quad\quad\quad\quad\hspace{10mm}\textbf{Then} Relabel vertices and edges if necessary to make $f_1\ne f_3$
   \vspace{0.mm}\item[15.]   \hspace{37mm}  {$\mathcal V_*\leftarrow \{v_i:i\equiv0\pmod3, i\in [k]-[3]\}$}
\vspace{0.mm}\item[16.]   \hspace{37mm}   $\textsc{Alg1}(\mathcal H)\leftarrow\{u_1,u_3\}\cup \mathcal V_*\cup\textsc{Alg1}(\mathcal H\setminus (\mathcal E_c\cup\{f_1,f_3\}))$
  \vspace{0.mm}\item[17.] \quad\quad\quad\quad\quad\hspace{9.5mm} \textbf{Else}  $\textsc{Alg1}(\mathcal H)\leftarrow\{u_2,u_4\}\cup  \textsc{Alg1}(\mathcal H\setminus (\mathcal E_c\cup\{f_1,f_2\}))$
    \vspace{0.mm}\item[18.] \quad\quad\quad\hspace{.7mm}  \textbf{If} $k\equiv2\pmod3$ 
  \vspace{0.mm}\item[19.] \quad\quad\quad\hspace{3mm}  {\textbf{Then}  $\textsc{Alg1}(\mathcal H)\leftarrow\{u_1\}\cup \{v_i:i\equiv1\pmod3, i\in [k]-[3]\}\cup\textsc{Alg1}(\mathcal H\setminus (\mathcal E_c\cup\{f_1\}))$}
  \vspace{0mm}\item[20.] Output $\textsc{Alg1}(\mathcal H)$
\vspace{-3mm}
\end{mylabel}
\caption{{Feedback Vertex Sets of Linear 3-Uniform Hypergraphs}} \label{alg1}
\end{algorithm}

 Note that Algorithm \ref{alg1} never visits isolated vertices {(it only scans along the edges of the current hypergraph)}.  The number of iterations performed by the algorithm is upper bounded by $|\mathcal E|$. {Since $\mathcal H$ is 3-uniform, the} condition in any step is checkable in $O(|\mathcal E|^2)$ time. Any cycle in {Step 7 or Step 9} can be found in $O(|\mathcal E|^2)$ time.\footnote{{The shortest path between any pair of vertices can be find in $O(|\mathcal E|)$ time using breadth first search. A shortest cycle can be find by checking all $O(|\mathcal E|)$ possibilities.}} Thus Algorithm \ref{alg1} runs in {$O(|\mathcal E|^3)$} time.

 \begin{corollary}\label{cor:1/3}
 Given any linear 3-uniform hypergraph $\mathcal H$, Algorithm \ref{alg1} finds in {$O(|\!|\mathcal H|\!|^3)$} time a FVS of $\mathcal H$ with size at most $|\!|\mathcal H|\!|/3$.\qed
 \end{corollary}

\begin{lemma}\label{acycle}
If $\mathcal H=(\mathcal V,\mathcal E)$  is a connected linear $3$-uniform hypergraph without cycles, then $|\mathcal V|=2|\mathcal E|+1$.
\end{lemma}

\begin{proof}
We prove by induction on $|\mathcal E|$. The base case where $|\mathcal E|=0$ is trivial. Inductively, we assume that {$|\mathcal E|\ge1$} and the
lemma holds for all connected  acyclic linear 3-uniform hypergraph of edges fewer than $\mathcal H$. Take arbitrary $e\in \mathcal E$. Since $\mathcal H$ is connected, acyclic and 3-uniform, $\mathcal H\setminus e$ contains exactly three components $\mathcal H_{i}=(\mathcal V_{i},\mathcal E_{i})$, $i=1,2,3$. Note that for each $i\in[3]$, hypergraph  $\mathcal H_i$ with $|\mathcal E_i|<|\mathcal E|$ is connected, linear, 3-uniform and acyclic. By the induction hypothesis, we have $|\mathcal V_{i}|=2|\mathcal E_{i}|+1$ for $i=1,2,3$. It follows that $|\mathcal V|=\sum_{i=1}^3|\mathcal V_{i}|= 2\sum_{i=1}^3|\mathcal E_{i}| +3= 2|\mathcal E|+1$.
\end{proof}

Given any hypergraph $\mathcal H=(\mathcal V,\mathcal E)$, we  can easily find a minimal (not necessarily minimum) FES in $O(|\mathcal E|^2)$ time: Go through the edges of the trivial FES $\mathcal E$ in any order, and remove the edge from  the FES immediately if the edge is redundant. The redundancy test can be implemented using Depth First Search.

\begin{lemma}\label{cyclecover}
Let $\mathcal H=(\mathcal V,\mathcal E)$ be a linear $3$-uniform hypergraph with $p$ components. If $\mathcal F$ is a minimal FES of $\mathcal H$, then $|\mathcal F|\leq 2|\mathcal E|-|\mathcal V|+p$. In particular, $\tau_c^{{}_{\mathcal E}}(\mathcal H)\le 2|\mathcal E|-|\mathcal V|+p$.
\end{lemma}

\begin{proof}  Suppose that $\mathcal H\setminus \mathcal F$ contains exactly $k$ components $\mathcal H_{i}=(\mathcal V_{i},\mathcal E_{i})$, $i=1,\ldots,k$. It follows from Lemma~\ref{acycle} that $|\mathcal V_{i}|=2|\mathcal E_{i}|+1$ for each $i\in [k]$. Thus $|\mathcal V|=\sum_{i\in [k]}|\mathcal V_{i}|= 2\sum_{i\in [k]}|\mathcal E_{i}|+k= 2(|\mathcal E|- |\mathcal F|)+k$, which means $2|\mathcal F|= 2|\mathcal E|-|\mathcal V|+k$. To establish the lemma, it suffices to prove  $k\leq |\mathcal F|+ p$.

 In case of  $|\mathcal F|=0$, we have $\mathcal F=\emptyset$ and $k=p=|\mathcal F|+p$.  In case of  $|\mathcal F| \ge1$, suppose that $\mathcal F=\{e_{1},...,e_{|\mathcal F|}\}$. Because $\mathcal F$ is a minimal FES of $\mathcal H$,  for each $i\in [|\mathcal F|]$, there is a cycle $\mathcal C_{i}$ in $\mathcal H\setminus (\mathcal F\setminus \{e_{i}\})$ such that $e_{i}\in \mathcal C_{i}$, and $\mathcal C_i\setminus e_i$ is a path in $\mathcal H\setminus \mathcal F$ connecting two of the three vertices in $e_i$. Considering $\mathcal H\setminus\mathcal F$ being obtained from $\mathcal H$ be removing $e_1,e_2,\ldots,e_{|\mathcal F|}$ sequentially, for $i=1,\ldots,|\mathcal F|$, since $|e_i|=3$, the presence of path $\mathcal C_i\setminus e_i$ implies that the removal of $e_i$ can create at most one more component. Therefore we have  $k\le p+|\mathcal F|$ as desired.
\end{proof}

Given a {hypergraph} $\mathcal H=(\mathcal V,\mathcal E)$ with $n$ vertices and $m$ edges, {let $M_{\mathcal H}$ be the $\mathcal V\times\mathcal E$ incidence matrix}. From $M_{\mathcal H}$, we may construct a bipartite graph $G_{\mathcal H}$ with bipartition $\mathcal V,\mathcal E$ such that there is an edge of $G_{\mathcal H}$ between $v\in\mathcal V$ and $e\in\mathcal E$ if and only if $v\in e$ in $\mathcal H$.

 Suppose that $\mathcal H$ is acyclic. It is easy to see that $G_{\mathcal H}$ is acyclic. Thus $M=M_{\mathcal H}$ falls within the class of {\em restricted totally unimodular} (RTUM) matrices  defined by Yannakakis \cite{yannakakis1985}. As the name indicates, RTUM matrices are all totally unimodular. Hence the total unimodularity and LP duality give the well-known result \cite{berge1989} that {$\tau(\mathcal H)=\min\{\mathbf 1^T\mathbf x:M^T\mathbf x\ge\mathbf 1,x\geq 0\}=\max\{\mathbf 1^T\mathbf y:M\mathbf y\le\mathbf 1,y\geq 0\}=\nu(\mathcal H)$}. Moreover, since $M$ is RTUM, both a minimum transversal and a maximum matching of $\mathcal H$ can be found in  $O(n(m+n\log n)\log n)$ time using Yanakakis's combinatorial algorithm \cite{yannakakis1985} based on the current best combinatorial algorithms for the $b$-matching problem and the maximum weighted independent set problem on a bipartite mulitgraph with $n$ vertices and $m$ edges, where the bipartite $b$-matching problem can be solved with the minimum-cost flow algorithm in $O(n\log n(m+n\log n))$ time (see Section 21.5 and Page 356 of \cite{schrijver2003}) and the maximum weighted independent set problem can be solved with maximum flow algorithm in $O(nm\log n)$ time (See Pages 300-301 of \cite{yannakakis1985}).

\begin{theorem}[\cite{berge1989,{yannakakis1985}}]\label{cyclefree}
Let $\mathcal H$ be a hypergraph with $n$ {non-isolated} vertices and $m$ edges. If $\mathcal H$ has no cycle, then $\tau(\mathcal H)=\nu(\mathcal H)$, and a minimum transversal and a maximum matching of $\mathcal H$ can be found in $O(n(m+n\log n)\log n)$ time.\qed
\end{theorem}


 \section{Triangle packing and covering}\label{sec:tuza}

{This section establish several new sufficient conditions for Conjecture \ref{coj:tuza} as well as their algorithmic implications on finding minimum triangle covers. Section \ref{sec:high} deals with graphs of high triangle packing numbers. Section \ref{sec:irr} investigates irreducible graphs with many edges. Section \ref{sec:er} discusses Erd\H{o}s-R\'{e}nyi graphs with high densities.}

To each  graph $G=(V,E)$,  we associate a hypergraph $\mathcal H_G=(E,\mathscr T_G)$,  referred to as {\em triangle hypergraph} of $G$, such that the  vertices and edges of $\mathcal H_G$  are the edges and    triangles of $G$, respectively. Since $G$ is simple, it is easy to see that  $\mathcal H_G $  is $3$-uniform and linear, $ \nu(\mathcal H_G)=\nu_{t}(G)$ and $\tau(\mathcal H_G)=\tau_{t}(G)$. Note that $|\!|\mathcal H_G|\!|=|\mathscr T_G|<\min\{|V|^3,|E|^3\}$, and {$|E|\le3|\mathscr T_G|$} if $G$ is irreducible{, i.e., $\cup_{T\in\mathscr T_G}E(T)=E$. Note that  the number of non-isolated vertices of $\mathcal H_G$ is upper bounded by $3|\!|\mathcal H_G|\!|=3|\mathscr T_G|$}.

\subsection{Graphs with many edge-disjoint triangles}\label{sec:high}
We investigate Tuza's conjecture for graphs with large packing numbers, which are firstly  compared  with the number of triangles, and then with the number of edges.
 \begin{theorem}\label{c}
If  graph  $G$ and real number $c\in (0,1]$ satisfy $\nu_{t}(G)/|\mathscr T_G|\ge c$, then a triangle cover of $G$ with size at most $ \frac{3c+1}{3c}\nu_t(G)$ can be found in {$O(|\mathscr T_G|^3)$} time, which implies $\tau_t(G)/\nu_t(G)\le\frac{3c+1}{3c}$.
 \end{theorem}

\begin{proof} 
{We}  consider the triangle hypergraph $\mathcal H_G=(E,\mathscr T_G)$ of $G$  {which} is 3-uniform and linear. By Corollary~\ref{cor:1/3}, we can  {find} in $O(|\mathscr T_G|^3)$ time a FVS $\mathcal S$ of $\mathcal H_G$ with $|\mathcal S|\le |\mathscr T_G|/3$.
  Since $\nu(\mathcal H_G)=\nu_{t}(G)\ge c|\mathscr T_G|$, it follows that $|\mathcal S|\le  \nu(\mathcal H_G)/(3c)$. As $\mathcal H_G\setminus S$ is acyclic, Theorem~\ref{cyclefree} enables us to find in $O(|\mathscr T_G|^2\log^2|\mathscr T_G|)$ time a minimum transversal $\mathcal R$ of $\mathcal H_G\setminus S$ such that   $|\mathcal R|=\tau(\mathcal H_G\setminus S)=\nu(\mathcal H_G\setminus \mathcal S)$. We observe that $\mathcal S\cup\mathcal R\subseteq E$ and $G\setminus (\mathcal S\cup\mathcal R)$ is triangle-free. Hence $\mathcal S\cup\mathcal R$ is a triangle cover of $G$ with size
  \[|\mathcal S\cup\mathcal R|
  \leq \frac{\nu(\mathcal H_G)}{3c}+\nu(\mathcal H_G\setminus\mathcal S)   \leq\frac{3c+1}{3c}\nu(\mathcal H_G)=\frac{3c+1}{3c}\nu_t(G),\] which proves the theorem.
\end{proof}

The special case of $c=1/3$ in the above theorem {gives the following result providing a new sufficient condition for Tuza's conjecture.}
 \begin{corollary}\label{th:1/3}If graph $G$   satisfies $\nu_{t}(G)/|\mathscr T_G|\ge1/3$, then $\tau_t(G)/\nu_t(G)\le2$.\qed
 \end{corollary}

The condition  $\nu_{t}(G)\ge|\mathscr T_G|/3$ in Corollary~\ref{th:1/3} applies, in some sense, only to the class of large scale  sparse graphs (which, e.g., does not include complete graphs on   four or more vertices). The mapping from the real number $c$ in the condition $\nu_{t}(G)\ge c|\mathscr T_G|$ to the coefficient $ \frac{3c+1}{3c}$ in  the conclusion  $\tau_t(G)\le \frac{3c+1}{3c}\nu_t(G)$ of Theorem \ref{c} shows  the trade-off between conditions and conclusions. As in {Corollary}~\ref{th:1/3}, $c = \frac13$  maps to  $ \frac{3c+1}{3c} = 2$ hitting the boundary  of   Tuza's conjecture. It remains to study graphs $G$ with $\nu_t(G)/|\mathscr T_G|<\frac13$. The next theorem (Theorem \ref{th:a-fourth}) tells us that actually we only need to take care of graphs $G$ with   $\nu_t(G)/|\mathscr T_G|\in(\frac14-\epsilon,\frac13)$, where $\epsilon$ can be any arbitrarily small positive number. So, in some sense, to solve Tuza's conjecture, we only have a gap of $\frac13-\frac14=\frac1{12}$ to be bridged. Interestingly, for $c = \frac14$, we have {$\frac{3c+1}{3c} = \frac73 = 2.333...$}, which is much better than the best known general bound 2.87 due to Haxell \cite{haxell19}. Only when  $c \leq \frac16$ does $\frac{3c+1}{3c}$ state a trivial bound equal to or greater than 3.

 \begin{theorem}\label{th:a-fourth}
If there exists some real $\delta>0$ such that Conjecture \ref{coj:tuza} holds for every graph $G$ with $\nu_t(G)/|\mathscr T_G|\ge1/4-\delta$, then Conjecture \ref{coj:tuza} holds for every graph.
 \end{theorem}

 \begin{proof} If $\delta\geq \frac14$, the theorem is trivial. We consider $0 < \delta < \frac14$. As the set of rational numbers is dense, we may assume $\delta\in\mathbb Q$ and   $ 1/4-\delta=i/j$ for some $i,j\in\mathbb N$. Therefore $i/j<1/4$ gives {$4i+1\le j$}, i.e., $4+1/i\le j/i$. It remains to prove that for any graph $G$ with $\nu_t(G)<(i/j)|\mathscr T_G|$ there holds $\tau_t(G)\le2\nu_t(G)$.

 Write $k$ for the positive integer $i|\mathscr T_G|-j\cdot\nu_t(G)$. Let $G'$ be the disjoint union of $G$ and $k$ copies of {$K_4$}. Clearly, $|\mathscr T_{G'}|=|\mathscr T_G|+k|\mathscr T_{K_4}|=|\mathscr T_G|+4k$, $\tau_t(G')=\tau_t(G)+k\cdot\tau_t(K_4)=\tau_t(G)+2k$ and $\nu_t(G')=\nu_t(G)+k\cdot\nu_t(K_4)=\nu_t(G)+k$.
 It follows that
\begin{eqnarray*}
(i/j)|\mathscr T_{G'}|&=&(i/j)(|\mathscr T_G|+4k)\\
&=&(i/j)((k+j\cdot\nu_t(G))/i+4k)\\
&=&(i/j)(j\cdot\nu_t(G)/i+(4+1/i)c)\\
&\le& \nu_t(G)+k\\
&=& \nu_t(G')
\end{eqnarray*}
where the inequality is guaranteed by $4+1/i\le j/i$. So $\nu_t(G')\ge(1/4-\delta)|\mathscr T_{G'}|$ together with the hypothesis of the theorem implies $\tau_t(G')\le2\nu_t(G')$, i.e., $\tau_t(G)+2k\le2(\nu_t(G)+k)$, giving $\tau_t(G)\le2\nu_t(G)$ as desired.
 \end{proof}

 \redcomment{In the proof of the above theorem, the property of $K_4$ that  $\nu_t(K_4)/|\mathscr T_{K_4}|=1/4$ and $\tau_t(K_4)/\nu_t(K_4)=2$ plays an important role. It helps to reduce the general Tuza's conjecture to the special case where $\nu_t(G)\ge(1/4-\delta)|\mathscr T_G|$.} 



The sufficient condition that compares the triangle packing number with the number of edges is based on the fact that every simple graph $G=(V,E)$ has a bipartite subgraph of at least $|E|/2$ edges, {which can be found in polynomial time}. Since this subgraph does not {contain} any triangle, we deduce that $\tau_t(G)\le|E|/2$, which implies the following result.

\begin{corollary}\label{cor:1/4}   If $G=(V,E)$  is a graph such that $\nu_{t}(G)/|E|\ge c$ for some $c>0$, then $\tau_t(G)/\nu_t(G)\le 1/(2c)$. In particular, if  $\nu_{t}(G)/|E|\ge1/4$, then $\tau_t(G)/\nu_t(G)\le 2$.\qed
\end{corollary}

{Thus if $\nu_{t}(G)/|E|\ge c$ for some $c>0$, then a triangle cover of $G$ with size at most $\nu_t(G)/(2c)$ can be found in polynomial time.} Complementary to Corollary \ref{th:1/3} whose condition mainly takes care of sparse graphs, the second statement of Corollary~\ref{cor:1/4} applies to many dense graphs, including complete graphs on $ 25$ or more vertices.

Similar to Corollary \ref{th:1/3} and  Theorem  \ref{th:a-fourth},   by which  our future investigation space on Tuza's conjecture shrinks to interval $(\frac14-\epsilon,\frac13)$ w.r.t. $\nu_t(G)/|\mathscr T_G|$, Corollary \ref{cor:1/4}   and the following Theorem \ref{th:1/5} narrow the interval w.r.t. $\nu_t(G)/|E|$   to $(\frac15-\epsilon,\frac14)$. Moreover, when taking $c = \frac15$ in Corollary \ref{cor:1/4}. we obtain $ \frac{1}{2c} = 2.5$, still better than Haxell's general bound $2.87$ \cite{haxell19}.

 \begin{theorem}\label{th:1/5}
If there exists some real $\delta>0$ such that Conjecture \ref{coj:tuza} holds for every graph $G$ with $\nu_t(G)/|E|\ge1/5-\delta $, then Conjecture \ref{coj:tuza} holds for every graph.
 \end{theorem}

 \begin{proof} We use the similar trick to that in proving Theorem \ref{th:a-fourth}; we add a number of complete graphs on five (instead of four) vertices. We may assume $\delta\in(0,\frac{1}{5})\cap \mathbb Q$ and   $ 1/5-\delta=i/j$ for some $i,j\in\mathbb N$. Therefore 
$i/j<1/5$ and the integrality of $i,j$ imply $5+1/i\le j/i$. To prove Tuza's conjecture for each graph $G$ with $\nu_t(G)<(i/j)|E|$, we write $k=i|E|-j\cdot\nu_t(G)\in\mathbb N$.  Let $G'=(V',E')$ be the disjoint union of $G$ and $k$ copies of  $K_5$'s. Then $|E'|=|E|+10k$, $\tau_t(G')=\tau_t(G)+k\cdot\tau_t(K_5)=\tau_t(G)+4k$, $\nu_t(G')=\nu_t(G)+k\cdot\nu_t(K_5)=\nu_t(G)+2k$, and
\[(i/j)|E'|=(i/j)(|E|+10k)
=(i/j)(j\cdot\nu_t(G)/i+(10+1/i)k)\le \nu_t(G)+2k= \nu_t(G')\]
where the inequality is guaranteed by $10+1/i\le 2j/i$. So $\nu_t(G')\ge(1/5-\delta)|E'|$ together with the hypothesis the theorem implies $\tau_t(G')\le2\nu_t(G')$, i.e., $\tau_t(G)+4k\le2(\nu_t(G)+2k)$, giving $\tau_t(G)\le2\nu_t(G)$ as desired.
 \end{proof}



\subsection{Graphs with many edges on triangles}\label{sec:irr}
Each graph has a unique maximum irreducible subgraph. Tuza's conjecture is valid for a graph if and only the conjecture is valid for its maximum irreducible subgraph. In this section, we study sufficient conditions for Tuza's conjecture on irreducible graphs that  bound the number of edges {below in terms of} the number of triangles.

\begin{theorem}\label{th:2}
If $G=(V,E)$ is an irreducible graph such that  $|E|/|\mathscr T_G|\ge2$, then a triangle cover of $G$ with cardinality at most $ 2\nu_t(G)$ can be found in $O(|\mathscr T_G|^2\log^2|\mathscr T_G|)$ time, which implies $\tau_t(G)/\nu_t(G)\le2$.
 \end{theorem}

\begin{proof}Suppose that the linear 3-uniform hypergraph $\mathcal H=(E,\mathscr T_G)$ associated to $G$ has exactly $p$ components. By Lemma~\ref{cyclecover}, we can find in $O(|\mathscr T_G|^2)$ time a minimal FES $\mathcal F$ of $\mathcal H$ such that   $|\mathcal F|\leq 2|\mathscr T_G|-|E|+p\le p$. Since $G $ is irreducible, we see that  $\mathcal H$ has no isolated vertices, i.e., every component of $\mathcal H$ has at least one edge. Thus  $\nu(\mathcal H)\geq p\geq |\mathcal F|$.   For the acyclic hypergraph $ \mathcal H\setminus\mathcal F$,  By Lemma~\ref{cyclefree} we may found in $O(|\mathscr T_G|^2\log^2|\mathscr T_G|)$ time a minimum transversal $\mathcal R$ of $\mathcal H\setminus \mathcal F$ such that   \[|\mathcal R|=\tau(\mathcal H\setminus\mathcal F)=\nu(\mathcal H\setminus\mathcal F).\] Observe that $\mathcal R\subseteq E$ and $\mathcal F\subseteq\mathscr T_G$. If $\mathcal F=\emptyset$, set $\mathcal S=\emptyset$, else for each $F\in\mathcal F$, take $e_F\in E$ with $e_F\in F$, and set $\mathcal S=\{e_F:F\in\mathcal F\}$. It is clear that $\mathcal R\cup\mathcal S$ is a transversal of $\mathcal H$ (i.e., a triangle cover of $G$) with cardinality $|\mathcal R\cup\mathcal S|\leq\nu(  \mathcal H\setminus\mathcal F)+ |\mathcal F| \leq 2\nu(\mathcal H)=2\nu_t(G)$, establishing the theorem.
 \end{proof}

We observe that the  graphs $G$ which consist of a number of triangles  sharing a common edge satisfy  $|E(G)|\ge 2|\mathscr T_G|$, but   {$\nu_{t}(G)< |\mathscr T_{G}|/3$} when $|\mathscr T_{G}|\geq 4$. So in  some sense, Theorem~\ref{th:2} works a supplement of Corollary \ref{th:1/3} for sparse graphs.

A multigraph is {\em series-parallel} if and only if it can be constructed
from a single edge by iteratively performing the {\em D-Operation} of doubling an edge and/or the {\em S-Operation} of subdividing an edge. 
{A graph is a {\em $2$-tree} if and only if it can be constructed from a single edge by iteratively performing the {\em DS-Operation} of doubling an edge and subdivide the new edge with a new vertex.}
A subgraph of a 2-tree is called a {\em partial $2$-tree}. It is well-known that a {(simple)} graph is a partial 2-tree if and only if  all of its maximal 2-connected subgraphs are series-parallel \cite{bodlaender1998}. Thus, a series-parallel  {(simple)}
graph is a partial 2-tree. {In the following we show that every partial 2-tree $G$ satisfies $|E(G)|\ge 2|\mathscr T_G|$.}
\begin{corollary}
If $G=(V,E)$ is a partial $2$-tree, then a triangle cover of  $G$ with cardinality at most $ 2\nu_t(G)$ can be found in $O(|E|^2\log^2|E|)$ time.
 \end{corollary}
\begin{proof} In $O(|E|^2)$ time, we may remove from $G$ all edges that are not contained in any triangles. The resulting graph is still a partial 2-tree. So we may assume without loss of generality that $G$ is irreducible. Since each triangle of $G$ is contained a unique maximal 2-connected subgraph of $G$, we may further assume that $G$ is 2-connected. It follows that $G$ is series-parallel. Since $G$ is simple, it can be constructed from a single edge by iteratively performing the S-Operation and/or the DS-Operation. The S-Operation  increases the number of edges and dose not change the number of triangles, while the DS-Operation   increases the number of edges  by 2 and the number of triangles by 1. Therefore, we have $|E|\ge 2|\mathscr T_G|$. The conclusion follows from Theorem \ref{th:2}.
\end{proof}

{Note that}  partial $2$-trees are $K_4$-free planar graphs. The validity of Tuza's conjecture on partial $2$-trees has been verified in \cite{tuza1990,HKT2012}. {The 2-approximation algorithm for finding a minimum triangle cover in planar graphs implied by Tuza's proof \cite{tuza1990} runs in $O(|E|)$ time.}

 Along the same line as in the previous subsection,  regarding Tuza's conjecture on graph $G$, Theorem \ref{th:2} and the  following Theorem \ref{th:3/2} jointly narrow the interval w.r.t. $|E(G)|/|\mathscr T_G|$ to $(1.5-\epsilon,2)$ for future study.
  \begin{theorem} \label{th:3/2}
If there exists some real $\delta>0$ such that Conjecture \ref{coj:tuza} holds for every irreducible  graph $G=(V,E)$  with $|E|/|\mathscr T_G|\ge3/2-\delta$, then Conjecture \ref{coj:tuza} holds for every irreducible graph (and therefore every graph).
 \end{theorem}

 \begin{proof} Again we apply the trick of adding copies of $K_4$. We may assume $\delta\in(0,3/2)\cap\mathbb Q$ and   $ 3/2-\delta=i/j$ for some $i,j\in\mathbb N$. Therefore $2i+1\le 3j$, implying $(i/j)(4+1/i)\le 6$.

 For any irreducible graph $G$ with $|E|<(i/j)|\mathscr T_G|$, we write $k=i|\mathscr T_G|-j|E|\in\mathbb N$. Let $G'$ be the disjoint union of $G$ and $k$ copies of $K_4$. 
 Then $G'$ is irreducible, and
  \[
(i/j)|\mathscr T_{G'}|=(i/j)(|\mathscr T_G|+4k)=(i/j)(j|E|/i+(4+1/i)k)\le |E|+6k= |E'|.\]
 It follows from the hypothesis of the theorem that $\tau_t(G')\le2\nu_t(G')$, i.e., $\tau_t(G)+2k\le2(\nu_t(G)+k)$, giving $\tau_t(G)\le2\nu_t(G)$ as desired.
 \end{proof}



\subsection{Erd\H{o}s-R\'{e}nyi graphs with high densities}\label{sec:er}

Let $n$ be a positive integer, and let $p\in[0,1]$. The Erd\H{o}s-R\'{e}nyi random graph model \cite{Alon2008} is a probability space over the set $\mathcal G(n,p)$ of graphs $G=(V,E)$ on the vertex set $V=\{1,...,n\}$, where an edge between vertices $i$ and $j$ is included in $E$ with probability $p$ independent from every other possible edge, i.e., \[\text{\bf Pr}[ij\in E]=p\text{ for each pair of distinct }i,j\in V.\] The $\mathcal G(n,p)$   model is often used in the probabilistic method for tackling problems in various areas such as graph theory and combinatorial optimization.

The following result on the triangle packing numbers of complete graphs \cite{Brualdi2009} is useful in deriving a good estimation for the triangle packing numbers of graphs in $\mathcal G(n,p)$.

\begin{theorem}[\cite{Brualdi2009}]\label{steiner}
$\nu_t(K_{n})=|E(K_{n})|/3$   if and only if $n\equiv1,3\pmod6$.\qed
 \end{theorem}

\begin{theorem}\label{random}
Suppose that $p> \sqrt{3}/2$ and $G=(V,E)\in\mathcal G(n,p)$. Then $\text{\bf Pr}[\nu_t(G)\geq |E|/4]=1-o(1)$ and $\text{\bf Pr}(\tau_t(G)\le2\nu_t(G))=1-o(1)$.
 \end{theorem}

\begin{proof} Let $K_n$ denote the complete graph on $V$. For each edge  $e\in K_{n}$, let $X_{e}$ be the indicator variable satisfying: $X_{e}=1$ if $e\in E$ and $X_{e}=0$ otherwise. Thus $\text{\bf E}[X_{e}]= p$, $X=\sum_{e\in K_{n}}X_{e}=|E|$, $ \text{\bf E}[X]=n(n-1)p/2$.  Since $X_{e}, e\in K_{n}$, are independent 0-1 variables, by Chernoff Bounds,  for each $\epsilon \in (0,1]$, $ \text{\bf Pr}[X> (1+\epsilon)\text{\bf E}[X]]\leq exp(-\epsilon^{2}\text{\bf E}[X]/3)=exp(-\epsilon^{2}n(n-1)p/6)=o(1)$. So  \[\text{\bf Pr}[X\leq (1+\epsilon)\text{\bf E}[X]]= \text{\bf Pr}(X\leq (1+\epsilon)n(n-1)p/2)
=1-o(1).\]

On the other  hand,   by Theorem~\ref{steiner},   we can make $K_n$ have an edge-disjoint triangle decomposition by  deleting  at most three vertices, which implies that $\nu_t(K_{n})$ is lower bounded by $k=\lceil(n-3)(n-4)/6\rceil$. Thus we can take $k$ edge-disjoint triangles $T_{1}, \ldots, T_k$ from $K_{n}$. For each $i\in [k]$, let  $Y_{i}$ be the indicator variable satisfying: $Y_{i}=1$ if $T_{i}\subseteq G $ and $Y_{i}=0$ otherwise. Note that $\text{\bf E}[Y_{i}]= p^{3}$ for each $i\in[k]$, $ \nu_t(G)\geq Y = \sum_{i=1}^kY_{i}$ and $ \text{\bf E}[Y]=kp^{3}$. Because $T_1,\ldots,T_k$ are edge-disjoint, $Y_1,\ldots,Y_k$ are independent 0-1 variables. By  Chernoff Bounds, for each $\epsilon \in (0,1)$, $\text{\bf Pr}[Y< (1-\epsilon)\text{\bf E}[Y]]\leq exp(-\epsilon^{2}\text{\bf E}[Y]/2)\leq exp(-\epsilon^{2}(n-3)(n-4)p^{3}/12)=o(1)$.Thus
\[\text{\bf Pr}[\nu_t(G)\geq (1-\epsilon)(n-3)(n-4)p^{3}/6]\ge\text{\bf Pr}[\nu_t(G)\geq (1-\epsilon)kp^3]\ge\text{\bf Pr}[Y\geq (1-\epsilon)\text{\bf E}[Y]]=1-o(1).\]

Recall that $p> \sqrt{3}/2$. We can take $\epsilon \in (0,1)$ such that $\lim_{n\rightarrow \infty}\frac{(1-\epsilon)(n-3)(n-4)p^{3}/6}{(1+\epsilon)n(n-1)p/8}= \frac{4p^{2}(1-\epsilon)}{3(1+\epsilon)}> 1$. So for sufficient large $n$, we always have $(1-\epsilon)(n-3)(n-4)p^{3}/6 > (1+\epsilon)n(n-1)p/8$. Since we have $\nu_t(G)\geq (1-\epsilon)(n-3)(n-4)p^{3}/6$   with probability $1-o(1)$ and have $|E|=X\leq (1+\epsilon)n(n-1)p/2$ with probability $1-o(1)$, we obtain $\nu_t(G)\geq |E|/4$ with probability $1-o(1)$. It follows from Corollary \ref{cor:1/4} that $\text{\bf Pr}(\tau_t(G)\le2\nu_t(G))=1-o(1)$.
\end{proof}

\section{Conclusion}\label{sec:conclude}
Using tools from hypergraphs, we design polynomial-time algorithms for finding {a} small triangle covers {in} graphs, which particularly imply several sufficient conditions for Tuza's conjecture {(Conjecture \ref{coj:tuza})}.
\paragraph{Triangle packing and covering.} In this paper, we have established new sufficient conditions $\nu_t(G)/|\mathscr T_G|\ge1/3$ and $|E|/|\mathscr T_G|\ge2$ for Tuza's conjecture on packing and covering triangles in graphs $G$. We prove the sufficiency by designing {polynomial-time} combinatorial algorithms for finding a triangle cover of $G$ whose cardinality is upper bounded by $2\nu_t(G)$. The high level {idea} of these algorithms is to remove {\em some edges} from $G$ so that the triangle hypergraph of the remaining graph is {\em acyclic} (see the proofs of Theorems \ref{1/3} and \ref{th:2}), which guarantees that the remaining graph has equal triangle covering number and triangle packing number, and a minimum triangle cover of the remaining graph is computable in polynomial time (see Theorem \ref{cyclefree}). It is well-known that the acyclic condition in Theorem \ref{cyclefree} could be weakened to odd-cycle-freeness \cite{yannakakis1985}. So the lower bound $1/3$ and $2$ in the sufficient conditions   could be (significantly) improved if we can remove (much) {\em fewer edges} from $G$ such that the triangle hypergraph of the remaining graph is {\em odd-cycle free}.

In view of   Theorems~\ref{th:a-fourth}, \ref{th:1/5} and \ref{th:3/2}, the study on the graphs $G$ satisfying $\nu_t(G) / |\mathscr T_G|\ge1/4 $ or $\nu_t(G)/ |E|\ge1/5$ or $|E|/|\mathscr T_G|\ge3/2$ might suggest more insight and foresight for resolving Tuza's conjecture. These graphs are critical in the sense that they are   standing on the border of the resolution.

\redcomment{Let us paying more attention to {\em extremal graphs} $G$ which satisfy Tuza's conjecture with tight ratio $\tau_t(G)/\nu_t(G)=2$. Actually, from Theorem~\ref{c}, Corollary~\ref{cor:1/4} and Theorem~\ref{th:2}, we can get a nice observation: for every irreducible extremal graph $G=(V,E)$, the following three inequalities hold on: $\nu_t(G)/  |\mathscr T_G|\leq 1/3$, $\nu_t(G)/ |E|\leq 1/4$, and $|E|/|\mathscr T_G|\leq 2$. Gregory J. Puleo first notices this observation.}

Another intermediate step towards resolving Tuza's conjecture is investigating its validity for  the classical Erd\H{o}s-R\'{e}nyi random  graph model $\mathcal G(n,p)$. In this paper, we have shown that Tuza's conjecture holds with high probability for    graphs in  $\mathcal G(n,p)$ when $p> \sqrt{3}/2$. It would be nice to prove the same result for {$p\in(0,\sqrt3/2]$}.

\paragraph{The generalization to {linear} 3-uniform hypergraphs.}
Our work has shown   very close relations between triangle packing and covering in graphs and edge (resp. cycle) packing and covering in linear 3-uniform hypergraphs. The theoretical and algorithmic results on linear 3-uniform hypergraphs  (Corollary \ref{cor:1/3} and Lemma \ref{cyclecover}) are crucial for us to establish sufficient conditions for Tuza's conjecture, and to find in strongly polynomial time a ``small'' triangle cover under the conditions (see Corollary \ref{th:1/3} and Theorem \ref{th:2}). Recall that, for any graph $G$, its triangle hypergraph $\mathcal H_G$ is linear 3-uniform, and Tuza's conjecture is equivalent to $\tau(\mathcal H_G)\leq 2\nu(\mathcal H_G)$.  \redcomment{As a natural generalization, one may ask: Does $\tau(\mathcal H)\leq 2\nu(\mathcal H)$ hold for all linear 3-uniform hypergraphs $\mathcal H$? It is easy to see that $\{\mathcal H_G:G$ is a graph\} is {\em properly} contained in the set of linear 3-uniform hypergraphs. Unfortunately, Zbigniew Lonc pointed out there is a simple negative example: The Fano projective plane is an example of a linear 3-uniform hypergraph whose matching number is 1 and transversal number is 3(See Figure~\ref{Fano}). 
 Last but not the least, the arguments in the paper have actually proved the following  stronger result.}

\begin{figure}[h]
\begin{center}
\includegraphics[scale=0.48]{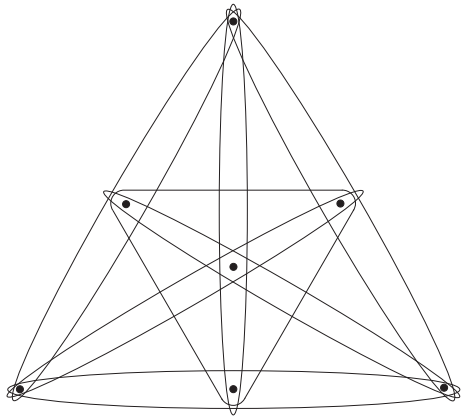}
\caption{\label{Fano}The Fano projective plane}
\end{center}
\end{figure}

\begin{theorem}\label{th:condition'}
Let  $\mathcal H=(\mathcal V,\mathcal E)$ be a linear $3$-uniform hypergraph without isolated vertices. Then a transversal of $\mathcal H$ with cardinality at most $2\nu(\mathcal H)$ can be found in polynomial time, which implies $\tau(\mathcal H)\le2\nu(\mathcal H)$, if one of the following conditions is satisfied:
(i) $\nu(\mathcal H)/|\mathcal E|\ge\frac13$, (ii) $|\mathcal V|/|\mathcal E|\ge2$.\qed
\end{theorem}

\redcomment{Comparing the above result on linear 3-uniform hypergraphs $\mathcal H$ with its counterpart on graphs presented in Theorem \ref{th:condition}, one might notice  that the condition on the lower bound of $\nu(\mathcal H)/|\mathcal V|$ is missing. This reason is that we do not have a nontrivial constant upper bound on  $\tau(\mathcal H)/|\mathcal V|$.}

\paragraph{Acknowledgements:} \redcomment{The authors are indebted to Gregory J. Puleo and Zbigniew Lonc for their invaluable comments and
suggestions which have greatly improved the presentation of this paper.}

\bibliography{ref}

\end{document}